\documentclass{article}
 
\usepackage[utf8]{inputenc}
\usepackage[english]{babel}

\advance\textheight by 2cm
\advance\topmargin -1cm

\advance\textwidth by 2cm
\advance\oddsidemargin by -1cm
\advance\evensidemargin by -1cm

\usepackage{changepage}
\usepackage{lipsum}
\usepackage{amssymb}
\usepackage{amsmath}
\usepackage{quoting}

\usepackage{mathtools}

\usepackage{tikz}
\usetikzlibrary{decorations.pathmorphing, decorations.pathreplacing, calc}
\usetikzlibrary{arrows.meta, positioning}

\usepackage{graphicx}
\usepackage{multicol}
\usepackage{footmisc}

\usepackage{centernot}
\usepackage{xcolor}

\usepackage{todonotes}

\newtheorem{theorem}{Theorem}[section]

\newtheorem{lemma}{Lemma}[section]
\newtheorem{definition}{Definition}[section]

\newtheorem{fact}{Fact}

\newcounter{noqed}
\newcommand{\qed}{ \ifmmode\mbox{ }\fi\rule[-.05em]{.3em}{.7em}\setcounter{noqed}{0}}
\newenvironment{proof}[1][{}]{\noindent{\bf Proof#1. }\setcounter{noqed}{1}}{\ifnum\value{noqed}=1\qed\fi\par\medskip}

\usepackage{amsmath}

\newcommand{\adj}{\,\text{\textemdash}\,}

\newcommand{\nott}[1]{\overline{#1}}
\newcommand{\hh}{^{\vphantom{-}}}
\newcommand{\primehh}{^{\prime\vphantom{-}}}

\title{Score and Rank Semi-Monotonicity for Closeness, Betweenness and Harmonic Centrality}
\author{\textbf{Paolo Boldi \;\; Davide D'Ascenzo \;\; Flavio Furia \;\; Sebastiano Vigna}}
\date{%
    Dipartimento di Informatica, Università degli Studi di Milano, Italy\\[2ex]%
    \today
}

\begin{document}
\maketitle

\begin{abstract}
    In the study of the behavior of centrality measures with respect to network
    modifications,
    \emph{score monotonicity} means that adding an arc increases the centrality score of the
    target of the arc; \emph{rank monotonicity} means that adding an arc improves the
    importance of the target of the arc relative to the remaining nodes. It is known~\cite{axioms,rank_mon} that score and rank monotonicity hold in directed graphs for
    almost all the classical centrality measures.
    In undirected graphs one expects that the corresponding properties (where both endpoints of the new edge enjoy the increase in score/rank) hold when adding a new edge. However,
    recent results~\cite{boldi_furia_vigna_2023} have shown that in undirected networks this is
    not true: for many centrality measures, it is possible to find situations where adding an edge reduces the rank of one of its two endpoints. In this paper we introduce a weaker condition for undirected networks, \emph{semi-monotonicity},
    in which just one of the endpoints of a new edge is required to enjoy score
    or rank monotonicity. We show that this condition is satisfied by closeness and betweenness centrality,
    and that harmonic centrality satisfies it in an even stronger sense.
\end{abstract}

\section{Introduction and Definitions}

In this paper we discuss the behavior of centrality measures in undirected networks after the addition of a new edge. In particular, we are interested
in the following question: if a new edge is added to a network, does the
importance of \emph{at least one} of its two endpoints increase?
This question was left open in~\cite{boldi_furia_vigna_2023}, where the authors proved that for many centrality measures it is possible to find situations where adding an edge reduces the rank of one of its two endpoints. Note that these results are in jarring contrast with the corresponding properties for directed networks, where it is known~\cite{axioms,rank_mon} that score and rank monotonicity hold for almost all centrality measures.

Formally, in this paper we introduce \emph{semi-monotonicity}, a weaker condition than monotonicity for undirected networks in which we require that \emph{at least one} endpoint of the new edge enjoys monotonicity.
Score semi-monotonicity, in particular, means that adding a new edge increases the score of at least one of the two endpoints:
\begin{definition}[Score semi-monotonicity]
    Given an undirected graph $G$, a centrality $c$ is said to be \emph{score semi-monotone on $G$} iff for every pair of non-adjacent vertices $x$ and $y$ we have that
    \begin{equation*}
        c_{G}(x) < c_{G'}(x)\quad{or}\quad c_{G}(y) < c_{G'}(y),
    \end{equation*}
    where $G'$ is the graph obtained adding the edge $x\adj y$ to $G$.
    We say that $c$ is \emph{score semi-monotone on a set of graphs} iff it is score semi-monotone on all the graphs from the set.
\end{definition}

As we already know from the directed case, a score increase does not imply that the rank relations between the two vertices involved in the new edge and the other vertices in the network remain unchanged. For this reason, rank monotonicity was introduced, where we require that every vertex that used to be dominated is still dominated after the addition of the new edge. Formally, the request for at least one of the two endpoints can be expressed as follows:
\begin{definition}[Rank semi-monotonicity]
    \label{def:rsemi}
    Given an undirected graph $G$, a centrality $c$ is said to be \emph{rank semi-monotone on $G$} iff for every pair of non-adjacent vertices $x$ and $y$ at least one of the following two statements holds:
    \begin{itemize}
        \item for all vertices $z\neq x,y$:
              \begin{align*}
                   & c_{G}(z) < c_{G}( x) \Rightarrow c_{G'}(z) < c_{G'}(x) \text{ and} \\
                   & c_{G}(z) = c_{G}(x) \Rightarrow c_{G'}(z) \leq c_{G'}(x),
              \end{align*}
        \item  for all vertices $z\neq x,y$:
              \begin{align*}
                   & c_{G}(z) < c_{G}(y) \Rightarrow c_{G'}(z) < c_{G'}(y) \text{ and} \\
                   & c_{G}(z) = c_{G}(y) \Rightarrow c_{G'}(z) \leq c_{G'}(y),
              \end{align*}
    \end{itemize}
    where $G'$ is the graph obtained adding the edge $x\adj y$ to $G$.
    We say that $c$ is \emph{rank semi-monotone on a set of graphs} iff it is rank semi-monotone on all the graphs from the set.
\end{definition}

In particular, we say that $c$ is rank semi-monotone at $x$ if the first statement holds, and rank semi-monotone at $y$ if the second statement holds (if both statements hold, $c$ is rank monotone).

\begin{definition}[Strict rank semi-monotonicity]
    \label{def:srsemi}
    Given an undirected graph $G$, a centrality $c$ is said to be \emph{strictly rank semi-monotone on $G$} iff for every pair of non-adjacent vertices $x$ and $y$ at least one of the following two statements holds:
    \begin{itemize}
        \item for all vertices $z\neq x,y$:
              $c_{G}(z) \leq c_{G}(x) \Rightarrow c_{G'}(z) < c_{G'}(x)$,
        \item  for all vertices $z\neq x,y$:
              $c_{G}(z) \leq c_{G}(y) \Rightarrow c_{G'}(z) < c_{G'}(y)$,
    \end{itemize}
    where $G'$ is the graph obtained adding the edge $x\adj y$ to $G$.
    We say that $c$ is \emph{strictly rank semi-monotone on a set of graphs} iff it is strictly rank semi-monotone on all the graphs from the set.
\end{definition}

Again, we say that it is strictly rank semi-monotone at $x$ if the first statement holds, and strictly rank semi-monotone at $y$ if the second statement holds (if both statements hold, $c$ is strictly rank monotone).

In the rest of the paper, we assume that we are given an undirected connected graph $G$, and two non-adjacent vertices $x,y\in N_{G}$; $G'$ will be the graph obtained by adding the edge $x\adj y$ to $G$.
From now on, $d_{uv}$ will refer to the distance (i.e., the length of a shortest path) between $u$ and $v$ in $G$ (i.e., before the $x\adj y$ addition), and $d'_{uv}$ will refer to the distance in $G'$ instead. In general, we will use the prime symbol to refer to any property or function of $G$ when translated to $G'$.

\section{Distances and Basins}

Geometric centrality measures~\cite{axioms} depend only on distances between
vertices. In the next two sections we are going to prove new results about the
semi-monotonicity of two geometric centrality measures---\emph{closeness centrality}~\cite{bavelas,bavelas2} and \emph{harmonic
    centrality}~\cite{beauchamp,axioms}. To understand the semi-monotonic
behavior of these centrality measures, we introduce a key notion:
\begin{definition}[Basin]
    Given an undirected graph $G$ and two
    non-adjacent vertices $x$ and $y$ we define the \emph{basin of $x$ (with respect to $y$)} $K_{xy}$ and the \emph{basin of $y$ (with respect to $x$) $K_{yx}$} as
    \begin{align*}
        K_{xy} \coloneqq & \{u\in N_{G}\;|\;d_{ux} \leq d_{uy}\} \\
        K_{yx} \coloneqq & \{u\in N_{G}\;|\;d_{uy} \leq d_{ux}\} \\
    \end{align*}
\end{definition}
That is, the basin of $x$ contains those vertices that are not farther from
$x$ than from $y$: see Figure~\ref{fig:closeness} for an example. Note that the vertices
that are equidistant from $x$ and $y$ are included in both basins, and the score of such vertices cannot
change in any geometric centrality when adding the edge $x\adj y$.

\begin{figure}
    \centering
    \begin{tikzpicture}
        \pgfdeclareimage{img}{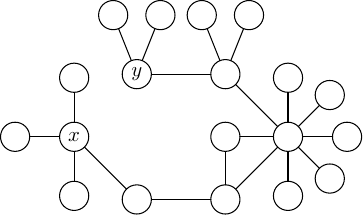};
        \node (img1) at (0,0) {\pgfuseimage{img}};
        \fill[blue!70,opacity=0.3] plot [smooth cycle, tension=0.5] coordinates {(-0.8,-0.2) (-1.8,2) (2.7,2) (3.5,-1) (3,-2) (1.4,-2.1) (1,-1) (0,-0.8)};
        \fill[red!70,opacity=0.3] plot [smooth cycle, tension=0.5] coordinates {(-0.2,-2) (-2,-2) (-3.6,-0.7) (-1.8,1) (-0.5,-0.8) (0.8,-0.1) (1.3,-1) (1,-2)};
    \end{tikzpicture}
    \caption{\label{fig:closeness}An undirected graph $G$, with $K_{xy}$ (the basin of x w.r.t.~y) shown in red and $K_{yx}$ (the basin of y w.r.t.~x) in blue.}
\end{figure}

Let us consider the following property:

\begin{definition}[Basin dominance]
    A centrality $c$ is said to be \emph{basin dominant} on an undirected graph $G$ iff for
    every pair of non-adjacent vertices $x$ and $y$ we have that
    \begin{gather}
        \begin{aligned}
            \label{def:basindom}
            c'(u)-c(u) \leq c'(x)-c(x) \qquad & \text{for every $u \in K_{xy}$, $u \neq x$}  \\
            c'(v)-c(v) \leq c'(y)-c(y) \qquad & \text{for every $v \in K_{yx}$, $v \neq y$.}
        \end{aligned}
    \end{gather}
    It is \emph{strictly basin dominant} iff the same conditions are satisfied, but inequalities (\ref{def:basindom}) hold with the $<$ sign.
\end{definition}
Intuitively, basin dominance means that the increase in score of $x$ and $y$ is at least as large as (or larger than, in the strict case) the increase in score of all other nodes in their respective basin.

The following theorems will be used throughout the paper:
\begin{theorem}
    \label{thm:generalbasin}
    If a centrality measure is strictly basin dominant on a graph then it is strictly rank semi-monotone on the same graph.
\end{theorem}
\begin{proof}
    Let $c$ be strictly basin dominant, and let us assume by contradiction that $c$ is not
    strictly rank semi-monotone. This implies that we should be able to find $u,v$ such that:
    \begin{equation}
        \label{eqn:absgeneral}
        \begin{cases}
            c(x) \geq c(v)   \\
            c(y) \geq c(u)   \\
            c'(v) \geq c'(x) \\
            c'(u) \geq c'(y).
        \end{cases}
    \end{equation}
    As a consequence of (\ref{eqn:absgeneral}), $c'(v)-c(v)\geq c'(x)-c(x)$ and $c'(u)-c(u)\geq c'(y)-c(y)$, which by the assumption of strict basin dominance imply $v \not\in K_{xy}$ and $u\not\in K_{yx}$. Therefore $v \in K_{yx}$ and $u \in K_{xy}$.
    But then again using strict basin dominance and (\ref{eqn:absgeneral}) we have $c'(y)-c(y)>c'(v)-c(v)\geq c'(x)-c(x)>c'(u)-c(u)\geq c'(y)-c(y)$, a contradiction.
    \qed\end{proof}

For the non-strict case, a similar result holds:
\begin{theorem}
    \label{thm:generalbasinw}
    If a centrality measure is basin dominant on a graph then it is rank semi-monotone on the same graph.
\end{theorem}
\begin{proof}
    Let $c$ be basin dominant, and let us assume by contradiction that $c$ is not rank semi-monotone. This implies that we should be able to find $u,v$ satisfying one
    of the following four sets of inequalities:
    \begin{equation}
        \label{eqn:basinloose}
        \begin{cases}
            c(x) > c(v)      \\
            c(y) > c(u)      \\
            c'(v) \geq c'(x) \\
            c'(u) \geq c'(y)
        \end{cases}
        \begin{cases}
            c(x) = c(v)   \\
            c(y) > c(u)   \\
            c'(v) > c'(x) \\
            c'(u) \geq c'(y)
        \end{cases}
        \begin{cases}
            c(x) > c(v)      \\
            c(y) = c(u)      \\
            c'(v) \geq c'(x) \\
            c'(u) > c'(y)
        \end{cases}
        \begin{cases}
            c(x) = c(v)   \\
            c(y) = c(u)   \\
            c'(v) > c'(x) \\
            c'(u) > c'(y).
        \end{cases}
    \end{equation}
    The proof is similar to the strict case.
    In all sets of inequalities we have $c'(v)-c(v)>c'(x)-c(x)$ and $c'(u)-c(u)>c'(y)-(y)$, which by the assumption of basin dominance imply $v \not\in K_{xy}$ and $u\not\in K_{yx}$. Therefore $v \in K_{yx}$ and $u \in K_{xy}$. Using basin dominance and (\ref{eqn:basinloose}), in all cases we have $c'(y)-c(y)\geq c'(v)-c(v)>c'(x)-c(x)\geq c'(u)-c(u)>$ $c'(y)-c(y)$, a contradiction.
    \qed\end{proof}

\section{Closeness Centrality}
Closeness centrality~\cite{bavelas,bavelas2} was shown to be score monotone but not rank monotone on connected undirected networks~\cite{boldi_furia_vigna_2023}. In the latter paper, it was left open the problem of whether closeness was (in our terminology) rank semi-monotone or not. In the rest of this section, we will solve this open problem by showing that closeness is in fact rank semi-monotone, but not in strict form.

Recall that, given an undirected connected graph $G=(N_{G},A_{G})$, the \emph{peripherality} of a vertex $u\in N_{G}$ is the sum of the distances from $u$ to all the other vertices of $G$:
\begin{equation*}
    p(u)=\sum_{v\in N_{G}}d_{uv}.
\end{equation*}
The \emph{closeness centrality} of $u$ is just the reciprocal of its peripherality:
\begin{equation*}
    c(u)=\frac{1}{p(u)}.
\end{equation*}
Consistently with the previous notation, we will use $p(u)$ for the peripherality of $u$ in $G$, and $p'(u)$ for the peripherality of $u$ in $G'$.

\begin{figure}
    \centering
    \begin{tikzpicture}[main/.style = {draw, circle}, node distance=2cm]
        \node[main] (y) {$y$};
        \node[main] (x) [right of=y] {$x$};
        \node[main] (z) [above of=y] {$z$};
        \node[main] (u) [right of=x] {$u$};

        \draw[decorate,decoration=snake] (z)--node[left]{$d_{yz}$}(y);
        \draw[decorate,decoration=snake] (z)--node[below]{$d_{xz}$}(x);
        \draw[decorate,decoration=snake] (z)--node[above]{$d_{uz}$}(u);
        \draw[decorate,decoration=snake] (y) to [bend right=30] node[below]{$d_{xy}$}(x);
        \draw[decorate,decoration=snake] (x)--node[below]{$d_{ux}$}(u);
        \draw[dashed] (x) -- (y);
    \end{tikzpicture}
    \caption{\label{fig:basin}Path labels represent the distance between the two endpoints.
        The dashed edge is the $x\adj y$ edge that we add to $G$, obtaining $G'$.}
\end{figure}
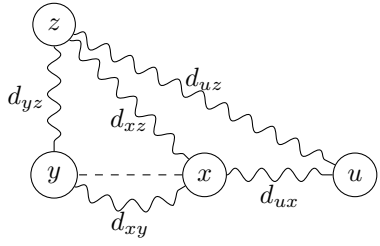

We have that:
\begin{lemma}\label{lemma:basin_closeness}
    Closeness centrality is basin dominant on connected undirected graphs.
\end{lemma}
\begin{proof}
    We first show that for every vertex $u \in K_{xy}$ and for every $z\neq u,x$:
    \begin{equation}
        \label{eqn:basin_closeness}
        d_{uz}-d'_{uz}\leq d_{xz}-d'_{xz}.
    \end{equation}
    Then, the result follows by adding both sides for all $z$ (as for $z=u$ or $z=x$
    the inequality trivializes) and using the fact that closeness is reciprocal to peripherality.
    To prove (\ref{eqn:basin_closeness}) we consider two cases (see Figure~\ref{fig:basin}):
    \begin{itemize}
        \item if $d'_{uz}=d_{uz}$ then the inequality trivially holds (because $d_{xz} \geq d'_{xz}$ always);
        \item if $d'_{uz}<d_{uz}$ then $d'_{uz}=d_{ux}+1+d_{yz}<d_{uz}$ and $d'_{xz}=d_{yz}+1$. Using $d_{uz}\leq d_{ux}+d_{xz}$ from the triangle inequality we have \begin{multline*}
                  d_{uz}-d'_{uz}=d_{uz}-(d_{ux}+1+d_{yz})\leq \\(d_{ux}+d_{xz})-(d_{ux}+1+d_{yz})=d_{xz}-(d_{yz}+1) = d_{xz}-d'_{xz}.\qed
              \end{multline*}
    \end{itemize}
    \end{proof}
This lemma is the undirected version of Lemma 2 in~\cite{rank_mon} (provided that in the latter you sum the inequality in the statement over all nodes $w$): it is interesting to observe that in the directed case the inequality holds \emph{for all $u$'s}, while here it holds only within the basin.

Applying Lemma~\ref{lemma:basin_closeness} and using Theorem~\ref{thm:generalbasinw}, we obtain that:
\begin{theorem}
    Closeness centrality is rank semi-monotone on connected undirected graphs.
\end{theorem}

Interestingly, and regardless of their initial score, we can always tell which of the two endpoints of the edge $x\adj y$ will have smaller peripherality (i.e., higher centrality) in $G'$.
In fact:
\begin{lemma}
    \label{lemma:perafter}
    The following property holds:
    \[
        p'(x)-p'(y)=|K_{yx}|-|K_{xy}|.
    \]
\end{lemma}
\begin{proof}
    We can write the peripherality of $x$ and $y$ in $G'$ as
    \begin{align*}
        p'(x) & =\sum_{u\in K_{xy}}d_{ux} + \sum_{u\in K_{yx}}(1+d_{uy}) - \sum_{u\in K_{xy} \cap K_{yx}}(1+d_{uy})  \\
        p'(y) & =\sum_{u\in K_{xy}}(1+d_{ux}) + \sum_{u\in K_{yx}}d_{uy} - \sum_{u\in K_{xy} \cap K_{yx}}(1+d_{ux}).
    \end{align*}
    Note that for each $u \in K_{xy} \cap K_{yx}$ we have $d_{ux}=d_{uy}$.
    Computing the difference between the two expressions gives the result.
    \qed\end{proof}

It is not hard to build a graph $G$ where $x$ has a smaller basin than $y$ but a greater score:  Lemma~\ref{lemma:perafter} tells us that $y$ becomes more central than $x$ in $G'$, due to having a greater basin.

\begin{figure}
    \centering
    \begin{tikzpicture}[main/.style = {draw, circle}, node distance=1cm, scale=.9]
        \node[main] (y1) at (0,4) {};
        \node[main] (y2) at (0,3) {};
        \node (yd) at (0,2) {$\vdots$};
        \node[main] (yk) at (0,1) {};
        \node[main] (y) at (2,2.5) {$y$};
        \node[main] (yru) at (3.5,3.5) {};
        \node[main] (yrd) at (3.5,1.5) {};
        \node[main] (x) at (5,2.5) {$x$};
        \node[main] (xru) at (6.5,3.5) {};
        \node[main] (xr) at (6.5,2.5) {};
        \node[main] (u) at (6.5,1.5) {$u$};
        \node[main] (w) at (8,2.5) {$w$};
        \node[main] (w1) at (10,4) {};
        \node[main] (w2) at (10,3) {};
        \node (wd) at (10,2) {$\vdots$};
        \node[main] (wk) at (10,1) {};

        \draw (y1) -- (y);
        \draw (y2) -- (y);
        \draw (yk) -- (y);
        \draw (y) -- (yru) -- (x);
        \draw (y) -- (yrd) -- (x);
        \draw (yru) -- (xru);
        \draw (x) -- (xru) -- (w);
        \draw (x) -- (xr) -- (w);
        \draw (x) -- (u) -- (w);
        \draw (w) -- (w1);
        \draw (w) -- (w2);
        \draw (w) -- (wk);

        \draw[dashed] (x) -- (y);

        \draw [decoration={brace,amplitude=0.5em},decorate,black]
        let \p1=(-0.5,4.1), \p2=(-0.5,0.9) in
        ({max(\x1,\x2)}, {\y2}) -- node[left=0.6em] {$k$} ({max(\x1,\x2)}, {\y1});

        \draw [decoration={brace,amplitude=0.5em},decorate,black]
        let \p1=(10.5,0.9), \p2=(10.5,4.1) in
        ({max(\x1,\x2)}, {\y2}) -- node[right=0.6em] {$k+4$} ({max(\x1,\x2)}, {\y1});

    \end{tikzpicture}
    \caption{\label{fig:closeness_strict_counterexample}A counterexample to strict rank semi-monotonicity for closeness centrality.
        For all $k\geq10$, $u$ and $x$ have the same score before and after the addition of the edge $x\adj y$. Moreover, $u$ has the same score of $y$ (or smaller) before the addition, but a higher score after the addition, breaking strict rank semi-monotonicity.}
\end{figure}
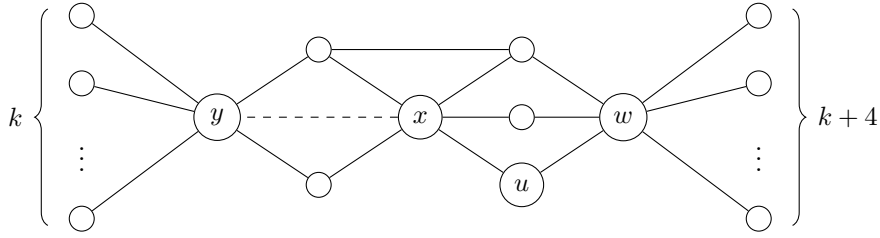

We conclude this section by showing that:
\begin{theorem}
    Closeness centrality is not strictly rank semi-monotone on (an infinite family of) connected undirected graphs.
\end{theorem}
\begin{proof}
    Consider the graphs in Figure~\ref{fig:closeness_strict_counterexample}, where $u\in K_{xy}$.
    This is an infinite family of graphs with a parameter $k$ which controls the sizes of the two stars around vertices $w$ and $y$.
    Computing the peripheralities of $u$, $x$ and $y$ before and after the addition of $x\adj y$, we obtain
    \begin{align*}
        p(u) & = 2\cdot(k+4) + 4\cdot k + 13  \qquad\qquad
        p'(u) = 2\cdot(k+4) + 3\cdot k + 12                    \\
        p(x) & = 3\cdot(k+4) + 3\cdot k + 9   \qquad\qquad\:\:
        p'(x) = 3\cdot(k+4) + 2\cdot k + 8                     \\
        p(y) & = 4\cdot(k+4) + k + 15    \qquad\qquad\quad\:
        p'(y) = 4\cdot(k+4) + k + 12.
    \end{align*}
    For all $k\geq10$, we have that
    \[
        p(x)=p(u),\quad p'(x)=p'(u),\quad p(y)\leq p(u),\quad p'(y)>p'(u),
    \]
    showing that closeness is not semi-monotone at $y$ (because $y$ used to be at least as central as $u$, but it is less central after the addition of the edge) and not strictly rank semi-monotone at $x$ (it is always as central as $x$, before and after adding the edge).
    \qed\end{proof}

\section{Harmonic Centrality}
Harmonic centrality~\cite{beauchamp} solves the issue of unreachable vertices in closeness centrality.
In particular, if we assume $\infty^{-1}=0$, we can define it as
\begin{equation*}
    h(u) = \sum_{v\in N_{G}\setminus\{u\}}\frac{1}{d_{uv}},
\end{equation*}
so that unreachable vertices have a null impact on the summation and, thus, on the final centrality score of the node.
Being a geometric measure, it is trivially score monotone but not rank monotone, as shown in~\cite{boldi_furia_vigna_2023},
where the same counterexample disproving rank monotonicity for closeness centrality also shows that harmonic centrality fails at satisfying this axiom.

\begin{lemma}\label{lemma:basin_harmonic}
    Harmonic centrality is strictly basin dominant on connected undirected graphs.
\end{lemma}
\begin{proof}
    We first show that for every vertex $u \in K_{xy}$ and for every $z\neq u,x$:
    \begin{equation}
        \label{eqn:harmineq}
        \frac{1}{d'_{uz}} - \frac{1}{d_{uz}} \leq \frac{1}{d'_{xz}} - \frac{1}{d_{xz}}.
    \end{equation}
    Then, the result follows by adding both sides for all $z$. Note that the unique term in $h(u)$ and $h'(u)$ (i.e., when $z = x$) is equal to the unique term in $h(x)$ and $h'(x)$  (i.e., when $z = u$), and they are both equal to 0. Also, we remark that (\ref{eqn:harmineq}) holds with a strict inequality at least in one case, i.e., when $z = y$, since $d'_{xy} < d_{xy}$ always.

    We consider two cases:
    \begin{itemize}
        \item if $d'_{uz}=d_{uz}$ then the inequality trivially holds (because $d'_{xz} \leq d_{xz}$ always);
        \item if $d'_{uz}<d_{uz}$ then $d'_{uz}=d_{ux}+1+d_{yz}<d_{uz}$ and $d'_{xz}=d_{yz}+1$. Using $d_{uz}\leq d_{ux}+d_{xz}$ from the triangle inequality we have
              \begin{align*}
                  \frac{1}{d'_{uz}} - \frac{1}{d_{uz}}-\Bigg(\frac{1}{d'_{xz}} - \frac{1}{d_{xz}}\Bigg) & =
                  \frac{1}{d_{ux}+1+d_{yz}} - \frac{1}{d_{uz}} - \Bigg(\frac{1}{d_{yz}+1} - \frac{1}{d_{xz}}\Bigg)                                                                            \\&= \frac{(d_{yz}+1)-(d_{ux}+1+d_{yz})}{(d_{ux}+1+d_{yz})(d_{yz}+1)} + \frac{d_{uz}-d_{xz}}{d_{uz}d_{xz}} \\
                                                                                                        & = - \frac{d_{ux}}{(d_{ux}+1+d_{yz})(d_{yz}+1)} + \frac{d_{uz}-d_{xz}}{d_{uz}d_{xz}} \\&< - \frac{d_{ux}}{d_{uz}d_{xz}} + \frac{(d_{ux}+d_{xz})-d_{xz}}{d_{uz}d_{xz}} = 0,
              \end{align*}
              which proves the inequality.\qed
    \end{itemize}
    \end{proof}

Using Lemma~\ref{lemma:basin_harmonic} and Theorem~\ref{thm:generalbasin}, we obtain that:
\begin{theorem}
    Harmonic centrality is strictly rank semi-monotone on connected undirected graphs.
\end{theorem}
The stronger result we can give for harmonic centrality should be compared to the fact that on
strongly connected graphs harmonic centrality is strictly rank monotone,
whereas closeness centrality is just rank monotone~\cite{boldi_furia_vigna_2023}.

\section{Betweenness Centrality}
Betweenness centrality~\cite{anthonisse,freeman}
focuses not only on the length of shortest paths, but also on how many of them involve a given node, trying
to estimate the amount of flow passing through nodes in a network. Note that betweenness is not a geometric measure.

Formally, if we call $\sigma_{vw}$ the number of shortest paths between two vertices $v$ and $w$ and $\sigma_{vw}(u)$ the number of such paths passing through $u$, then we can define the betweenness centrality of a vertex $u\in N_{G}$ as
\begin{equation*}
    b(u)=\sum_{\substack{i,j\neq u,\\\sigma_{ij}>0}}\frac{\sigma_{ij}(u)}{\sigma_{ij}}.
\end{equation*}

In the following, we let $N_G(u)$ denote the set of neighbors of $u$ in $G$, and $G[u]$ the subgraph of $G$ induced by $N_G(u)$ (sometimes called the \emph{ego network} of $u$).
As in the previous section, we denote with $\sigma$ and $\sigma'$ the number of shortest paths before and after the addition of an edge $x\adj y$, and with $b$ and $b'$
the betweenness centrality before and after the addition of the edge.

We know from~\cite{boldi_furia_vigna_2023} that this centrality measure is neither rank nor score monotone.
Nonetheless, we can show that the betweenness centrality of two vertices can never decrease after we link them with a new edge.
In fact:
\begin{lemma}
    \label{lemma:betposxy}
    The following properties hold:
    \iffalse
        \begin{align*}
            \frac{\sigma'_{ij}(x)}{\sigma'_{ij}} - \frac{\sigma_{ij}(x)}{\sigma_{ij}}  \geq 0 \qquad & \text{for all $i,j \neq x$} \\
            \frac{\sigma'_{ij}(y)}{\sigma'_{ij}} - \frac{\sigma_{ij}(y)}{\sigma_{ij}}  \geq 0 \qquad & \text{for all $i,j \neq y$}
        \end{align*}
        As a consequence
        \[
            b'(x) \geq b(x) \qquad\text{ and }\qquad b'(y)\geq b(y).
        \]
    \else
        \[
            \frac{\sigma'_{ij}(x)}{\sigma'_{ij}} - \frac{\sigma_{ij}(x)}{\sigma_{ij}}  \geq 0 \qquad  \text{for all $i,j \neq x$}.
        \]
        As a consequence, $b'(x) \geq b(x)$.
    \fi
\end{lemma}
\begin{proof}
    For all $i,j\in N_{G}$ such that $i,j\neq x$, let us call $p\hh_{\hh x}$ ($p\hh_{\hh \nott{x}}$, respectively) the number of shortest paths between $i$ and $j$ passing (not passing, resp.) through $x$.
    We have to show that, for each such pair $i,j\neq x$, the following holds:
    \begin{equation}\label{eqn:betposxy}
        \frac{\sigma'_{ij}(x)}{\sigma'_{ij}} - \frac{\sigma_{ij}(x)}{\sigma_{ij}}=
        \frac{p\primehh_{\hh x}}{p\primehh_{\hh x}+p\primehh_{\hh \nott{x}}} - \frac{p\hh_{\hh x}}{p\hh_{\hh x}+p\hh_{\hh \nott{x}}} \geq 0.
    \end{equation}
    Summing over all $i,j\neq x$ proves the second part of the statement.

    To show that (\ref{eqn:betposxy}) is indeed true we consider two cases:
    \begin{itemize}
        \item $d'_{ij}<d_{ij}$, meaning that in $G'$ all the shortest paths $i{\sim}j$ pass through the edge $x\adj y$ (in particular through $x$). Thus, we obtain:
              \[
                  1-\frac{p\hh_{\hh x}}{{p\hh_{\hh x}+p\hh_{\hh \nott{x}}}}\geq 0,
              \]
              which is clearly true, since the second term of the left side of the inequality is a value between $0$ and $1$.
        \item $d'_{ij}=d_{ij}$, which implies that $p\primehh_{\hh \nott{x}}=p\hh_{\hh \nott{x}}$ and $p\primehh_{\hh x}\geq p\hh_{\hh x}$.
              We express $p\primehh_{\hh x}$ as $p\hh_{\hh x}+\alpha$ with $\alpha\geq 0$, obtaining:
              \begin{align*}
                  \frac{p\hh_{\hh x}+\alpha}{p\hh_{\hh x}+\alpha+p\hh_{\hh \nott{x}}} - \frac{p\hh_{\hh x}}{p\hh_{\hh x}+p\hh_{\hh \nott{x}}} & =
                  \frac{p^{2}_{x}+p\hh_{\hh x}p\hh_{\hh \nott{x}}+\alpha p\hh_{\hh x}+\alpha p\hh_{\hh \nott{x}} - (p^{2}_{x}+\alpha p\hh_{\hh x}+p\hh_{\hh x}p\hh_{\hh \nott{x}})}{(p\hh_{\hh x}+\alpha+p\hh_{\hh \nott{x}})(p\hh_{\hh x}+p\hh_{\hh \nott{x}})} =       \\
                                                                                                                                              & = \frac{\alpha p\hh_{\hh \nott{x}}}{(p\hh_{\hh x}+\alpha+p\hh_{\hh \nott{x}})(p\hh_{\hh x}+p\hh_{\hh \nott{x}})} \geq 0,
              \end{align*}
              which is again true, concluding the proof.\qed
    \end{itemize}
\end{proof}

Incidentally, we can always tell if the betweenness centrality of a vertex is zero without actually computing it. In fact,
\begin{lemma}\label{lemma:betwenness_zero}
    Let $G$ be a connected undirected graph and $u\in N_{G}$.
    Then:
    \begin{equation*}
        b(u)=0\iff G[u] \text{ is a clique}.
    \end{equation*}
\end{lemma}
\begin{proof}
    If $b(u)=0$ no shortest paths are passing through $u$: but then any two neighbors of $u$ must be adjacent (or otherwise they would have distance $2$, and the path through $u$ has length $2$).
    Conversely, suppose that $G[u]$ is a clique and let $i,j \neq u$. A path from $i$ to $j$ cannot involve $u$, otherwise it would touch two neighbors of $u$, say $i',j'$, and we might shorten it by skipping $u$ and taking the $i'\adj j'$ edge instead of $i'\adj u\adj j'$.
    \qed\end{proof}

\begin{figure}
    \centering
    \begin{tikzpicture}[main/.style = {draw, circle}, node distance=2cm, scale=0.9]
        \foreach \n in {1,...,4} {
                \node[main] (v\n) at (45+\n*90:1.3) {};
            }
        \foreach \i in {1,...,4} {
                \foreach \j in {\i,...,4} {
                        \draw (v\i) -- (v\j);
                    }
            }
        \node[main] (x) at (270:2.6) {$x$};
        \node[main] (y) at (360:2.6) {$y$};
        \node[main] (u) at (180:2.6) {$u$};
        \foreach \i in {1,...,4} {
                \draw (x) -- (v\i);
                \draw (y) -- (v\i);
                \draw (u) -- (v\i);
            }
        \draw[dashed] (x) -- (y);
    \end{tikzpicture}
    \caption{\label{fig:betweenness_counterexample}Simple counterexample for score semi-monotonicity and strict rank semi-monotonicity for betweenness centrality.
        The dashed edge is the $x\adj y$ edge that we add to $G$, obtaining $G'$.
        The betweenness score of vertices $x$, $y$ and $u$ is $0$ both in $G$ and $G'$.}
\end{figure}
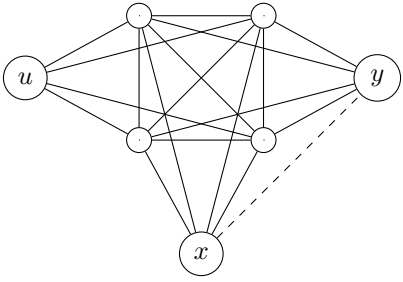

As a consequence, and differently from geometric measures, we can show that
\begin{theorem}\label{theorem:betweenness_score}
    Betweenness centrality is not score semi-monotone on (an infinite family of) connected undirected graphs.
\end{theorem}
\begin{proof}
    Consider a graph $G$ such that $G[x]$ and $G[y]$ are cliques and moreover, $N_{G}(x)=N_{G}(y)$, that is, $x$ and $y$ have the same neighborhood in $G$ (see Figure~\ref{fig:betweenness_counterexample} for an example).
    Then, by Lemma~\ref{lemma:betwenness_zero} we know that $b(x)=b(y)=0$.
    It is easy to observe that $G'[x]$ and $G'[y]$ are still cliques, hence we can use the same lemma and say that $b'(x)=b'(y)=0$, meaning that the addition of the $x\adj y$ leaves the score of both the endpoints unchanged.
    \qed\end{proof}

Moreover, we can say that
\begin{theorem}\label{theorem:betwenness_strict_rank}
    Betweenness centrality is not strictly rank semi-monotone on (an infinite family of) connected undirected graphs.
\end{theorem}
\begin{proof}
    Consider a graph $G$ with three vertices $x$, $y$ and $u$ adjacent to a clique, but with $x$, $y$, and $u$ not adjacent to each other, as in Figure~\ref{fig:betweenness_counterexample}.
    Then, by Lemma~\ref{lemma:betwenness_zero} we know that $b(x)=b(y)=b(u)=0$; if we add the edge $x \adj y$ to $G$, obtaining $G'$, we also have
    $b'(x)=b'(y)=b'(u)=0$ by the same lemma.
    \qed\end{proof}

We are now going to show that, somehow unexpectedly, betweenness centrality is rank semi-monotone on connected undirected graphs.
In fact, we show that it enjoys the same dominance property of closeness centrality, in spite of being a non-geometric measure:

\begin{lemma}\label{lemma:basin_betweenness}
    Betweenness centrality is basin dominant on connected undirected graphs.
\end{lemma}
\begin{proof}
    Let us call $\Delta_{z}=b'(z)-b(z)$ the score difference for a vertex $z\in N_{G}$, and for every pair of nodes $i,j$ (with $i\neq j$) let also
    \[
        \Delta_{z}(i,j)=\frac{\sigma'_{ij}(z)}{\sigma'_{ij}}-\frac{\sigma_{ij}(z)}{\sigma_{ij}}.
    \]
    Obviously
    \[
        \Delta_{z}=\sum_{i,j \neq z} \Delta_{z}(i,j).
    \]
    We want to show that $\Delta_u\leq\Delta_x$ for every $u \in K_{xy}$.
        The two summations happen on a different set of pairs of indices, and we treat the common and non-common pairs separately.

        The easiest case is that of pairs $i,j$ that appear in the summation of $\Delta_x$ but not in the summation
        of $\Delta_u$, because we know from Lemma~\ref{lemma:betposxy} that those summands are non-negative.

        Then we consider the pairs $i,j$ that appear in the summation of $\Delta_u$ but not in the summation of $\Delta_x$, that is,
        those where either $i$ or $j$ are equal to $x$; without loss of generality let us assume $j=x$. We want to show that
        in this case, instead, we have
        \[
            \Delta_u(i,x) = \frac{\sigma'_{ix}(u)}{\sigma'_{ix}} - \frac{\sigma_{ix}(u)}{\sigma_{ix}}\leq 0.
        \]
        When $i\in K_{xy}$, the new $x\adj y$ edge does not create any new shortest path between $i$ and $x$, so $\Delta_u(i,x) =0$.
        Conversely, when $i\not\in K_{xy}$
        new shortest paths cannot pass through $u$ (remember that $u \in K_{xy}$); thus,
        \begin{itemize}
            \item if $d_{ix} = d'_{ix}$ then $\sigma_{ix}\leq \sigma'_{ix}$ and
                  $\sigma_{ix}(u)=\sigma'_{ix}(u)$;
            \item if $d_{ix} > d'_{ix}$ then $\sigma'_{ix}(u) = 0$.
        \end{itemize}
        We are now left with the pairs $i,j$ that appear in both summations, that is, $i,j\neq u$ and $i,j\neq x$.
        In this case, we want to prove a term-by-term bound, that is:
        \begin{equation}
            \label{eqn:ineqij}
            \Delta_{u}(i,j)=\frac{\sigma'_{ij}(u)}{\sigma'_{ij}} - \frac{\sigma_{ij}(u)}{\sigma_{ij}} \leq \frac{\sigma'_{ij}(x)}{\sigma'_{ij}} - \frac{\sigma_{ij}(x)}{\sigma_{ij}}=\Delta_{x}(i,j).
        \end{equation}

        Note that if $i$ and $j$ belong to the same basin the edge $x\adj y$ does not create any new shortest
        path between $i$ and $j$, so (\ref{eqn:ineqij}) holds because both sides are zero;
        this includes the case in which $i$ and $j$ are in the intersection of the basins,
        so we can
        assume, without loss of generality, that $i\in K_{xy}\setminus K_{yx}$ and $j\in K_{yx}\setminus K_{xy}$. Note also that if $u$ is
        equidistant from $x$ and $y$, the edge $x \adj y$ will never create any shortest path between $i$
        and $j$ that passes through $u$. Hence, we will restrict our attention to the case where $u$
        is strictly closer to $x$ than to $y$, that is, $u\in K_{xy} \setminus K_{yx}$.

        The case $\Delta_u(i,j) \leq 0$ is trivial because of Lemma~\ref{lemma:betposxy}; we now analyze the case $\Delta_u(i,j) > 0$.


        Since  $i,u\in K_{xy} \setminus K_{yx}$ and $j\in K_{yx}\setminus K_{xy}$,
        all
        shortest paths in $G$ and $G'$ from $i$ to $j$ passing through $x$ and $y$
        must pass through $x$ before $y$, and all
        shortest path in $G$ and $G'$ from $i$ to $j$ passing through $u$ and $x$
        must pass through $u$ before $x$.

        The first statement is trivial (just look at the basins of $i$ and $j$). For the second
        statement, we can prove it by contradiction: if $u$ is after $x$ in a shortest path from $i$ to $j$, then $u$ is necessarily between $x$ and $y$. This means that $d'_{ij}<d_{ij}$ since we might shorten the path $x{\sim}u{\sim}y$ by taking $x\adj y$. Hence $\sigma'_{ij}(u)=0$ and thus
        $\Delta_u(i,j) < 0$---a contradiction.

        Let us now call $p\hh_{\hh ux}$ the number of shortest paths $i{\sim}j$ in $G$ passing through $u$ and then through $x$, $p\hh_{\hh \nott{u}x}$ the ones passing through $x$ but not $u$, $p\hh_{\hh u\nott{x}}$ the ones passing through $u$ but not $x$ and, finally, $p\hh_{\hh \nott{u}\nott{x}}$ the ones passing through neither of them. The same notations are used for $G'$, but we use $p'$ instead of $p$.

        With these notations, we can write the single terms appearing in~(\ref{eqn:ineqij}) as follows:
        \begin{align*}
            \sigma_{ij}(u)  & = p\hh_{\hh ux}+p\hh_{\hh u\nott{x}}  \qquad
            \sigma_{ij}(x)   = p\hh_{\hh ux}+p\hh_{\hh \nott{u}x}                                                                       \\
            \sigma'_{ij}(u) & = p\primehh_{\hh ux}+p\primehh_{\hh u\nott{x}}  \qquad
            \sigma'_{ij}(x)   = p' _{ux}+p\primehh_{\hh \nott{u}x}                                                                      \\
            \sigma_{ij}     & = p\hh_{\hh ux}+p\hh_{\hh \nott{u}x}+p\hh_{\hh u\nott{x}}+p\hh_{\hh \nott{u}\nott{x}}                     \\
            \sigma'_{ij}    & = p\primehh_{\hh ux}+p\primehh_{\hh \nott{u}x}+p\primehh_{\hh u\nott{x}}+p\primehh_{\hh \nott{u}\nott{x}}
        \end{align*}
        which makes us able to rewrite~(\ref{eqn:ineqij}) as
        \[
            \frac{p\primehh_{\hh ux}+p\primehh_{\hh u\nott{x}}}{\sigma'_{ij}}
            - \frac{p\hh_{\hh ux}+p\hh_{\hh u\nott{x}}}{\sigma_{ij}} \leq
            \frac{p\primehh_{\hh ux}+p\primehh_{\hh \nott{u}x}}{\sigma'_{ij}}
            - \frac{p\hh_{\hh ux}+p\hh_{\hh \nott{u}x}}{\sigma_{ij}}
        \]
        which is equivalent to
        \begin{align*}
            \frac{p\primehh_{\hh u\nott{x}} - p\primehh_{\hh \nott{u}x}}{\sigma'_{ij}} & \leq \frac{p\hh_{\hh u\nott{x}} - p\hh_{\hh \nott{u}x}}{\sigma_{ij}}.
        \end{align*}
        Hence,
        \begin{align*}
            \frac{
            (
            p\primehh_{\hh u\nott{x}}p\hh_{\hh ux} + p\primehh_{\hh u\nott{x}}p\hh_{\hh \nott{u}x} + p\primehh_{\hh u\nott{x}}p\hh_{\hh u\nott{x}} + p\primehh_{\hh u\nott{x}}p\hh_{\hh \nott{u}\nott{x}}
            ) - (
            p\primehh_{\hh \nott{u}x}p\hh_{\hh ux} + p\primehh_{\hh \nott{u}x}p\hh_{\hh \nott{u}x} + p\primehh_{\hh \nott{u}x}p\hh_{\hh u\nott{x}} + p\primehh_{\hh \nott{u}x}p\hh_{\hh \nott{u}\nott{x}}
            )}
            {\sigma'_{ij}\sigma_{ij}} \leq \\ \\
            \leq \frac{
            (
            p\primehh_{\hh ux}p\hh_{\hh u\nott{x}} + p\primehh_{\hh \nott{u}x}p\hh_{\hh u\nott{x}} + p\primehh_{\hh u\nott{x}}p\hh_{\hh u\nott{x}} + p\primehh_{\hh \nott{u}\nott{x}}p\hh_{\hh u\nott{x}}
            ) - (
            p\primehh_{\hh ux}p\hh_{\hh \nott{u}x} + p\primehh_{\hh \nott{u}x}p\hh_{\hh \nott{u}x} + p\primehh_{\hh u\nott{x}}p\hh_{\hh \nott{u}x} + p\primehh_{\hh \nott{u}\nott{x}}p\hh_{\hh \nott{u}x}
            )}
            {{\sigma'_{ij}\sigma_{ij}}}.
        \end{align*}
        As a consequence,
        \begin{align*}
            p\primehh_{\hh u\nott{x}}p\hh_{\hh ux} + p\primehh_{\hh u\nott{x}}p\hh_{\hh \nott{u}x} + p\primehh_{\hh u\nott{x}}p\hh_{\hh u\nott{x}} + p\primehh_{\hh u\nott{x}}p\hh_{\hh \nott{u}\nott{x}} + p\primehh_{\hh ux}p\hh_{\hh \nott{u}x} + p\primehh_{\hh \nott{u}x}p\hh_{\hh \nott{u}x} + p\primehh_{\hh u\nott{x}}p\hh_{\hh \nott{u}x} + p\primehh_{\hh \nott{u}\nott{x}}p\hh_{\hh \nott{u}x} \leq \\
            \leq p\primehh_{\hh ux}p\hh_{\hh u\nott{x}} + p\primehh_{\hh \nott{u}x}p\hh_{\hh u\nott{x}} + p\primehh_{\hh u\nott{x}}p\hh_{\hh u\nott{x}} + p\primehh_{\hh \nott{u}\nott{x}}p\hh_{\hh u\nott{x}} + p\primehh_{\hh \nott{u}x}p\hh_{\hh ux} + p\primehh_{\hh \nott{u}x}p\hh_{\hh \nott{u}x} + p\primehh_{\hh \nott{u}x}p\hh_{\hh u\nott{x}} + p\primehh_{\hh \nott{u}x}p\hh_{\hh \nott{u}\nott{x}}.
        \end{align*}
        Finally, deleting from both sides the two common summands $p\primehh_{\hh u\nott{x}}p\hh_{\hh u\nott{x}}$ and $p\primehh_{\hh \nott{u}x}p\hh_{\hh \nott{u}x}$, we obtain
        \begin{align*}
            p\primehh_{\hh u\nott{x}}p\hh_{\hh ux} + p\primehh_{\hh u\nott{x}}p\hh_{\hh \nott{u}x} + p\primehh_{\hh u\nott{x}}p\hh_{\hh \nott{u}\nott{x}} + p\primehh_{\hh ux}p\hh_{\hh \nott{u}x} + p\primehh_{\hh u\nott{x}}p\hh_{\hh \nott{u}x} + p\primehh_{\hh \nott{u}\nott{x}}p\hh_{\hh \nott{u}x} \leq \\
            \leq p\primehh_{\hh ux}p\hh_{\hh u\nott{x}} + p\primehh_{\hh \nott{u}x}p\hh_{\hh u\nott{x}} + p\primehh_{\hh \nott{u}\nott{x}}p\hh_{\hh u\nott{x}} + p\primehh_{\hh \nott{u}x}p\hh_{\hh ux} + p\primehh_{\hh \nott{u}x}p\hh_{\hh u\nott{x}} + p\primehh_{\hh \nott{u}x}p\hh_{\hh \nott{u}\nott{x}}.
        \end{align*}

        We now prove that:

    \begin{quoting}
		\begin{fact}
			\label{fact:pp}
			$p\primehh_{\hh ux}p\hh_{\hh \nott{u}x}=p\primehh_{\hh \nott{u}x}p\hh_{\hh ux}$.
		\end{fact}
		\begin{proof}[ of Fact~\ref{fact:pp}]
			\noindent{\emph{Case 1.}} If $p\primehh_{\hh ux}= 0$ (resp. $p\primehh_{\hh \nott{u}x} = 0$) then $p\hh_{\hh ux} = 0$ (resp. $p\hh_{\hh \nott{u}x} = 0$), since the edge $x\adj y$ can only shorten the paths passing through $x$; hence if $p\primehh_{\hh ux}= 0$ or $p\primehh_{\hh \nott{u}x} = 0$, the statement holds. 
			
			\noindent{\emph{Case 2.}} So let us assume that $p\primehh_{\hh ux} \neq 0$ and $p\primehh_{\hh \nott{u}x} \neq 0$. We want to show that in this case $p\hh_{\hh ux} \neq 0 \iff p\hh_{\hh \nott{u}x} \neq 0$. Let us denote with $d_{ix}^{\nott{u}}$ the length of the shortest path between $i$ and $x$ in $G$ that does not pass through $u$ (which is finite, because of the assumption $p\primehh_{\hh \nott{u}x} \neq 0$). We have:
			\begin{align*}
				p\primehh_{\hh ux}          \neq 0 &\iff d'_{ij} = d_{iu} + d_{ux} + d'_{xj},\\
				p\primehh_{\hh \nott{u}x}   \neq 0 &\iff d'_{ij} = d_{ix}^{\nott{u}}  + d'_{xj},\\
				p\hh_{\hh ux}               \neq 0 &\iff d_{ij}  = d_{iu} + d_{ux} + d_{xj},\\
				p\hh_{\hh \nott{u}x}        \neq 0 &\iff d_{ij}  = d_{ix}^{\nott{u}}  + d_{xj}.
			\end{align*}
			From the first two equations we have that $d_{iu} + d_{ux} = d_{ix}^{\nott{u}}$. Substituting in the third and fourth equation, we have that $p\hh_{\hh ux} \neq 0 \iff p\hh_{\hh \nott{u}x} \neq 0$.
			Hence, if $p\hh_{\hh ux} = 0$ also $p\hh_{\hh \nott{u}x} = 0$, and viceversa, making the equality in the statement true.

			\noindent{\emph{Case 3.}} 
			We are left with the case $p\primehh_{\hh ux} \neq 0$, $p\primehh_{\hh \nott{u}x} \neq 0$, $p\hh_{\hh ux} \neq 0$ and $p\hh_{\hh \nott{u}x} \neq 0$. Let $s\hh_{\hh u}$ ($s\hh_{\hh \nott{u}}$, respectively) be the number of shortest paths
			in $G$ from $i$ to $x$ that pass through $u$ (do not pass through $u$, resp.): this number remains the same in $G'$ because $i$ and $u$ both belong to the basin of $x$; let also $t$ and $t'$ be the number of shortest paths
			from $x$ to $j$ in $G$ and $G'$, respectively. Then,
			\[
			p\primehh_{\hh ux}p\hh_{\hh \nott{u}x} = s\hh_{\hh u} t' s\hh_{\hh \nott{u}} t =  s\hh_{\hh \nott{u}} t's\hh_{\hh u} t = p\primehh_{\hh \nott{u}x}p\hh_{\hh ux}.
			\]
			This concludes the proof of Fact~\ref{fact:pp}.\qed
		\end{proof}
    \end{quoting}

        Using Fact~\ref{fact:pp}, we can delete $p\primehh_{\hh ux}p\hh_{\hh \nott{u}x}$ from the left-hand side and $p\primehh_{\hh \nott{u}x}p\hh_{\hh ux}$ from the right-hand side, obtaining:
        \begin{align*}
            p\primehh_{\hh u\nott{x}}p\hh_{\hh ux} + p\primehh_{\hh u\nott{x}}p\hh_{\hh \nott{u}x} + p\primehh_{\hh u\nott{x}}p\hh_{\hh \nott{u}\nott{x}} + p\primehh_{\hh u\nott{x}}p\hh_{\hh \nott{u}x} + p\primehh_{\hh \nott{u}\nott{x}}p\hh_{\hh \nott{u}x} \leq \\
            \leq p\primehh_{\hh ux}p\hh_{\hh u\nott{x}} + p\primehh_{\hh \nott{u}x}p\hh_{\hh u\nott{x}} + p\primehh_{\hh \nott{u}\nott{x}}p\hh_{\hh u\nott{x}} + p\primehh_{\hh \nott{u}x}p\hh_{\hh u\nott{x}} + p\primehh_{\hh \nott{u}x}p\hh_{\hh \nott{u}\nott{x}}.
        \end{align*}

        We now distinguish two cases:
        \begin{itemize}
            \item If $d'_{ij}<d_{ij}$, all the shortest paths $i{\sim}j$ in $G'$ pass through $x$, meaning that $p\primehh_{\hh u\nott{x}}=p\primehh_{\hh \nott{u}\nott{x}}=0$.
                  Thus, we obtain
                  \[
                      0 \leq p\primehh_{\hh ux}p\hh_{\hh u\nott{x}} + p\primehh_{\hh \nott{u}x}p\hh_{\hh u\nott{x}} + p\primehh_{\hh \nott{u}x}p\hh_{\hh u\nott{x}} + p\primehh_{\hh \nott{u}x}p\hh_{\hh \nott{u}\nott{x}},
                  \]
                  which is always true since every term on the right is non-negative.
            \item If $d_{ij}=d'_{ij}$, all the existing shortest paths $i{\sim}j$ remain, and more can be created, meaning that
                  \begin{equation}\label{eqn:secondcasebet}
                      p\primehh_{\hh u\nott{x}}=p\hh_{\hh u\nott{x}},\quad
                      p\primehh_{\hh \nott{u}\nott{x}}=p\hh_{\hh \nott{u}\nott{x}},\quad
                      p\primehh_{\hh ux}\geq p\hh_{\hh ux},\quad
                      p\primehh_{\hh \nott{u}x}\geq p\hh_{\hh \nott{u}x}.
                  \end{equation}
                  Substituting the first two equalities, we must prove that:
                  \begin{align*}
                      p\hh_{\hh u\nott{x}}p\hh_{\hh ux} + p\hh_{\hh u\nott{x}}p\hh_{\hh \nott{u}x} + p\hh_{\hh u\nott{x}}p\hh_{\hh \nott{u}\nott{x}} + p\hh_{\hh u\nott{x}}p\hh_{\hh \nott{u}x} + p\hh_{\hh \nott{u}\nott{x}}p\hh_{\hh \nott{u}x} \leq \\
                      \leq p\primehh_{\hh ux}p\hh_{\hh u\nott{x}} + p\primehh_{\hh \nott{u}x}p\hh_{\hh u\nott{x}} + p\hh_{\hh \nott{u}\nott{x}}p\hh_{\hh u\nott{x}} + p\primehh_{\hh \nott{u}x}p\hh_{\hh u\nott{x}} + p\primehh_{\hh \nott{u}x}p\hh_{\hh \nott{u}\nott{x}},
                  \end{align*}
                  which is equivalent to
                  \[
                      0 \leq p\hh_{\hh u\nott{x}}(p\primehh_{\hh ux}-p\hh_{\hh ux}) + 2p\hh_{\hh u\nott{x}}(p\primehh_{\hh \nott{u}x}-p\hh_{\hh \nott{u}x}) + p\hh_{\hh \nott{u}\nott{x}}(p\primehh_{\hh \nott{u}x}-p\hh_{\hh \nott{u}x}),
                  \]
                  which is again true, since all the terms in parenthesis are non-negative, because of (\ref{eqn:secondcasebet}).\qed
        \end{itemize}
    \end{proof}

Hence, applying Theorem~\ref{thm:generalbasinw} with Lemma~\ref{lemma:basin_betweenness} we have that:
\begin{theorem}\label{theorem:betwenness_rank}
    Betweenness centrality is rank semi-monotone on connected undirected graphs.
\end{theorem}

\section{Conclusions and Future Work}

Table~\ref{tab:summ} summarizes the results of this paper along with those of~\cite{rank_mon,boldi_furia_vigna_2023,axioms}. For all the negative results, we have an infinite family of counterexamples (for instance, there are infinitely many graphs on which closeness is shown to be not strictly semi-monotone).

The notion of basin dominance turned out to be the key idea in all proofs of semi-monotonicity. It would be interesting to investigate whether
basin dominance applies to other geometric measures, or even other centrality measures based on shortest paths, as in
that case one gets immediately rank semi-monoto\-nicity.

Proving or disproving score and (strict) rank semi-monotonicity for other measures (in particular, for the spectral ones) remains an open problem.

\begin{table}[h]
    \centering
    {\scriptsize
        \begin{tabular}{l||l|l||l|l}
                                & \multicolumn{2}{c||}{undirected}       & \multicolumn{2}{c}{directed~\cite{rank_mon,axioms}}                                                         \\\hline
                                & \multicolumn{1}{c|}{score}             & \multicolumn{1}{c||}{rank}                          & \multicolumn{1}{c|}{score} & \multicolumn{1}{c}{rank} \\\hline
            Closeness           & monotone~\cite{boldi_furia_vigna_2023} & {\bf semi-monotone}                                 & monotone                   & monotone                 \\
            Harmonic centrality & monotone~\cite{boldi_furia_vigna_2023} & {\bf strictly semi-mon}.                            & monotone                   & strictly monotone        \\
            Betweenness         & {\bf not semi-monotone}                & {\bf semi-monotone}                                 & not monotone               & not monotone             \\
        \end{tabular}}
    \vspace*{3mm}
    \caption{\label{tab:summ}Summary of the results about monotonicity obtained in this paper (in boldface) and in~\cite{rank_mon,boldi_furia_vigna_2023,axioms}. All results are about (strongly) connected graphs.}
\end{table}

\section*{Acknowledgments}
This work was supported in part by project SERICS (PE00000014) under the NRRP MUR program funded by the EU - NGEU. Davide D'Ascenzo has been financially supported by the Italian National PhD Program in Artificial Intelligence (DM 351 intervento M4C1 - Inv. 4.1 - Ricerca PNRR), funded by EU - NGEU.

\bibliography{bibliography}

\begin{thebibliography}{1}
\providecommand{\url}[1]{{#1}}
\providecommand{\urlprefix}{URL }
\expandafter\ifx\csname urlstyle\endcsname\relax
  \providecommand{\doi}[1]{DOI~\discretionary{}{}{}#1}\else
  \providecommand{\doi}{DOI~\discretionary{}{}{}\begingroup
  \urlstyle{rm}\Url}\fi

\bibitem{anthonisse}
Anthonisse, J.M.: The rush in a directed graph.
\newblock Journal of Computational Physics pp. 1--10 (1971)

\bibitem{bavelas}
Bavelas, A.: A mathematical model for group structures.
\newblock Human Organization \textbf{7}, 16--30 (1948)

\bibitem{bavelas2}
Bavelas, A.: Communication patterns in task-oriented groups.
\newblock Journal of the Acoustical Society of America \textbf{22}, 725--730
  (1950)

\bibitem{beauchamp}
Beauchamp, M.A.: An improved index of centrality.
\newblock Behavioral Science \textbf{10}(2), 161--163 (1965)

\bibitem{boldi2023score}
Boldi, P., D'Ascenzo, D., Furia, F., Vigna, S.: Score and rank
  semi-monotonicity for closeness, betweenness and harmonic centrality (2023)

\bibitem{boldi_furia_vigna_2023}
Boldi, P., Furia, F., Vigna, S.: Monotonicity in undirected networks.
\newblock Network Science p. 1–23 (2023).
\newblock \doi{10.1017/nws.2022.42}

\bibitem{rank_mon}
Boldi, P., Luongo, A., Vigna, S.: Rank monotonicity in centrality measures.
\newblock Network Science \textbf{5}(4), 529–550 (2017).
\newblock \doi{10.1017/nws.2017.21}

\bibitem{axioms}
Boldi, P., Vigna, S.: Axioms for centrality.
\newblock Internet Mathematics \textbf{10}, 222 -- 262 (2014)

\bibitem{freeman}
Freeman, L.C.: Centrality in social networks conceptual clarification.
\newblock Social Networks \textbf{1}, 215--239 (1978)

\end{thebibliography}

\end{document}